\documentclass[12pt]{article}

\usepackage{amsmath}
\usepackage{graphicx}
\usepackage{amsfonts}
\usepackage{color}
\usepackage{amssymb}
\newtheorem{defin}{\sc {Definition}}[section]
\newtheorem{definition}[defin]{\sc {Definition}}
\newtheorem{theorem}[defin]{\sc {Theorem}}
\newtheorem{lemma}[defin]{\sc {Lemma}}

\newtheorem{proposition}[defin]{\sc {Proposition}}

\newtheorem{remark}[defin]{\sc {Remark}}
\newtheorem{example}[defin]{\sc {Example}}
\newenvironment{proof}[1][Proof]{\textbf{#1:} }{\ \rule{0.5em}{0.5em}}
\newcommand{\numbercellong}[2]
{
\begin{picture}(80,20)(0,0)
\put(0,0){\framebox(80,20)} \put(40,10){\makebox(0,0){#1}}
\end{picture}
}

\newcommand{\ep}{\varepsilon}

\newcommand{\dR}{{{\bf R}}}
\newcommand{\E}{{{\bf E}}}

\begin{document}
\title{Cooperation under Incomplete Information on the Discount Factors%
\thanks{The work of Solan was supported by the Israel Science Foundation, Grant \#212/09.}}
\author{Cy Maor\thanks{Institute of Mathematics, The Hebrew University, Jerusalem 91904, Israel.
              e-mail: cy.maor@mail.huji.ac.il}  and Eilon Solan\thanks{The School of Mathematical Sciences,
Tel Aviv University, Tel Aviv 69978, Israel. e-mail:
eilons@post.tau.ac.il}}

\maketitle

\begin{abstract}
In repeated games, cooperation is possible in
equilibrium only if players are sufficiently patient, and long-term gains from cooperation outweigh
short-term gains from deviation. What happens if the players have
incomplete information regarding each other's discount factors? In this
paper we look at repeated games in which each player has incomplete
information regarding the other player's discount factor, and ask
when full cooperation can arise in equilibrium. We provide
necessary and sufficient conditions that allow full cooperation in
equilibrium that is composed of grim trigger strategies,
and characterize the states of the world in which full cooperation occurs.
We then ask whether these
``cooperation events'' are close to those in the complete
information case, when the information on the other player's discount factor is ``almost'' complete.
\end{abstract}

\section{Introduction}

Cooperation is an important theme in human behavior.
In repeated interactions, cooperation is often achieved by threats:
if a player deviates from an agreed upon plan,
the other players will punish him.
The effectiveness of the threat is determined by the loss that the deviator incurs if the punishment is executed:
if the loss is high, the threat is effective.

Controlling behavior by threats is possible only if the players assign high enough weight to future payoffs,
that is, if they are sufficiently patient, and future payoffs sufficiently affect their overall utility.
This implies that sometimes cooperation cannot arise in equilibrium:
if a player is supposed to cooperate but is not patient, he will prefer to deviate and obtain short-term profits.
The same will happen if a player believes that some other player is not patient;
in this case, the player will believe that that player will not cooperate,
and therefore he will have no reason to cooperate as well.

In the present paper we study a model in which each of two players does not
know how patient the other player is; that is, the players have
incomplete information regarding each other's discount factor.
We ask whether and when full cooperation can arise in equilibrium,
or, more exactly, what
should be the players' beliefs so that cooperation arises in equilibrium.
To develop our ideas we consider the repeated Prisoner's Dilemma (see
Figure 1), which is a classical game that is used to study issues
relating to cooperation.

\centerline{
\begin{picture}(215,70)(-15,0)
\put(-15,8){Cooperate}
\put( 5,28){Defect}
\put( 65,50){Defect}
\put(140,50){Cooperate}
\put( 40, 0){\numbercellong{$0,4$}{}}
\put(40,20){\numbercellong{$1,1$}{}}
\put(120,0){\numbercellong{$3,3$}{}}
\put(120,20){\numbercellong{$4,0$}{}}
\end{picture}}
\label{figure1}

\smallskip

\centerline{Figure 1: The Prisoner's Dilemma}
\bigskip

In this game, both players have two actions, Cooperate ($C$) and Defect ($D$).
The action $D$ strongly dominates the action $C$ in the one-shot game.
If the players wish to cooperate, they may want to play $(C,C)$ at every stage,
thereby obtaining the payoff $(3,3)$ at every stage.
By deviating to $D$, a player profits 1 at the stage of deviation.
This deviation will not be profitable if the total loss that the player will incur due to this deviation is higher than 1.

There are many ways in which a player can punish the other player for deviating.
One way uses the \emph{grim trigger} strategy: if the other player deviates at stage $t$ and plays $D$ instead of $C$,
punish him by playing $D$ in all future stages.
If a player employs the grim trigger strategy,
then the deviator loses at least 2 in every stage following a deviation.
In particular, as soon as the discount factor of the deviator is at least $\frac{1}{3}$,
the deviation is not profitable.
Another way to punish a deviator is by using the \emph{tit-for-tat} strategy, in which at every stage the player plays the action
that the other player played in the previous stage.
It turns out that this punishment deters deviations as soon as the discount factor of the deviator is at least $\frac{1}{3}$ as well.

These constructions yield equilibrium when the players' discount factors
are common knowledge.
In this case, each player knows how much the other player stands to lose for deviating,
and therefore he knows whether the other player will be willing to cooperate.
In the present paper we study two-player repeated games
in which the players have incomplete information regarding each other's discount factor.
Our goal is to study when full cooperation may arise in equilibrium;
that is, what are the beliefs of the players
that allow play paths in which the players cooperate all through the game.

To illustrate the complexity of the analysis of this problem,
suppose that the discount factors of the players are not common knowledge.
A naive strategy profile is one in which at the first stage each player signals
whether or not he is willing to cooperate, and in subsequent stages
the players cooperate only if both expressed willingness to cooperate at the first stage.
Signalling at the first stage is achieved by playing $C$ if one's discount factor is at least $1\over3$,
and playing $D$ otherwise.
This profile may not be an equilibrium,
because if a player's discount factor is above $1\over3$,
but he assigns high probability to the event that his opponent's discount factor is lower than $1\over3$,
then the player has incentive to deviate and play $D$ at the first stage.

As mentioned before, it is easier to support cooperation when the punishment is severe.
Our goal is to focus on the conditions on the beliefs of a player
that are needed for cooperation in equilibrium.
We will therefore consider
the most severe punishment, which is achieved by the grim trigger
strategy; that is, if a player deviates, he is punished in all
subsequent stages. We thus allow each player to be of one of two
types: a non-cooperative type who always defects, and a
cooperative type who follows the grim trigger strategy. As we show
through an example (see Example \ref{example new}), the grim
trigger strategy is not necessarily the strategy that yields the
most widespread cooperation, that is, cooperation in the largest
set of states of the world. We choose to focus on the grim trigger
strategy because, due to its simplicity, it allows us to focus on
the players' beliefs, even under complex information structures.

Under the grim trigger strategy profile, if the player's discount factor is below $1\over3$, then he will not cooperate,
and will be of the non-cooperative type.
The player will also be of the non-cooperative type if he assigns sufficiently high probability
to the event that his opponent is of the non-cooperative type.
He will also refuse to cooperate if he assigns high probability to the event that the opponent
assigns high probability to him being of the non-cooperative type, etc.
Thus, cooperation can arise only when both players' discount factor is at least $1\over3$,
each player assigns sufficiently high probability to the event that the other player's discount factor is
at least $1\over3$,
each player assigns sufficiently high probability to the event that the other player assigns sufficiently
high probability to the event that his (the player's) discount factor is at least $1\over3$, etc.
In other words, an infinite list of requirements must hold so that full cooperation will arise.

Not surprisingly, the infinite list of requirements
necessary to support an equilibrium can be summarized by two conditions.
So that full cooperation arises in equilibrium one should require that
(a) each player's discount factor is at least $1\over3$,
and (b) each player assigns a sufficiently high probability that the other player is going to cooperate.
For the conditions to be sufficient, we need to add a third condition, (c) that
whenever a player does not cooperate, he assigns a sufficiently high probability to the event that the other player is not going to cooperate.
Interestingly, the probability in (b) is not constant, but depends on the player's
discount factor:
the higher his
discount factor, the more the player will lose if cooperation is not achieved.
Therefore, a player with high
discount factor will cooperate in situations where he wouldn't
have cooperated if his discount factor were lower.

To model incomplete information on the discount factor
we use a Harsaniy's game with incomplete information,
where the state of nature is the pair of the players' discount factors, and the concept of Bayesian equilibrium.%
\footnote{The results would not change if we used perfect Bayesian
equilibria instead of Bayesian equilibria.}
A pair of events $(K_1,K_2)$ is a pair of \emph{cooperation events}
if the strategy pair in which each player $i$ plays the grim
trigger strategy on $K_i$ and always defects on the complement of
$K_i$ is a Bayesian equilibrium. That is, when cooperation events
occur, full cooperation is possible in equilibrium.

We start by providing a complete characterization of pairs of cooperation events.
We then define
a new concept of $f$-belief, which is a generalization of Monderer and Samet's (1989)
concept of $p$-belief:
when $f$ is a real-valued function defined over the set of states of the world,
an event $A$ is an $f$-belief of player $i$ at the state of the world $\omega$
if the conditional probability that player $i$ assigns to $A$ at $\omega$ is at least $f(\omega)$.
This concept is closely related to the cooperation condition (b) mentioned above,
since in order to cooperate each player needs to assign a sufficiently high probability that the other player is going to cooperate,
a probability that depends on his own discount factor, and therefore on the state of the world;
that is, the player has to $f$-believe that the other player is going to cooperate, for a certain function $f$.
Using the notions of iterated $f$-belief and common $f$-belief we provide an iterative construction
of the largest pair of cooperation events.

The characterization of cooperation events can be
generalized for other two-player games. For $2\times 2$ games the
results remain fairly the same, with an appropriate
adjustments of conditions (a)--(c).
In the case of larger games, while (the equivalent of) conditions (a)--(c) still characterize cooperation events,
our construction of the largest pair of cooperation events is no longer valid.

Another solution concept that was studied in the literature
is that of \emph{interim correlated rationalizabiliy} (ICR), see, e.g., Dekel, Fudenberg, and Morris (2007) and Weinstein and Yildiz (2007).
With respect to this concept the results have the same flavor, yet condition (c) is not necessary for a grim trigger strategy to be rationalizable.
In this case the construction of cooperation events via $f$-beliefs is simpler and can be easily extended to games larger than $2\times 2$.

A natural question is whether, when there is almost complete information on the discount factors,
there are cooperation events that are close to the cooperation events in the case of complete information.
Monderer and Samet (1989) ensure that there is an equilibrium that coincide with the complete information equilibrium in most states of the world,
when the information is almost complete in a certain sense.
However, in their setting there are always ``unmapped areas'' in this almost complete information equilibrium;
that is, states of the world in which the strategy is not specified.
We ask whether, in an almost complete information setting, there is an equilibrium in
 grim trigger strategies that is defined in every state of the world and is close to the complete information equilibrium.
It turns out that the answer depends on the definition of almost complete information.
We will provide two natural definitions for this concept;
in one, similar to Mondrer and Samet (1989),
the answer is positive for the prisoner's dilemma and a certain class of other games,
but generally it is negative, while in the other it is positive for all games.

Our analysis demonstrates the significance of the concept of $f$-belief to the study of cooperation
in the presence of incomplete information.
Plainly the decision whether or not to cooperate depends on the whole hierarchy of belief of the player;
our analysis reveals that it crucially depends on
the exact discount factor of the player, and not only,
say, on whether it is above or below some fixed threshold, as is the case in the games with complete information.
Finally, our analysis shows that the details of the definition of almost complete information
drastically affect the continuity of the largest pair of cooperation events as a function of the beliefs of the players,
and we point out at one possible definition of the concept that guarantees the continuity of this function.

The literature regarding repeated games with different discount factors, and specifically with incomplete information regarding them,
is quite scant.
The most thorough analysis of repeated games with different discount factors with \emph{complete} information that we are aware of is
Lehrer and Pauzner (1999) who characterize the equilibrium payoffs in two-player games and show that
the set of feasible payoffs in the repeated game is typically larger than the convex hull of the underlying stage-game payoffs.
Lehrer and Yariv (1999) analyze the case of two-player zero-sum repeated game with one-sided incomplete information regarding the payoff matrix
in which the discount factors are common knowledge.
The analysis closest to ours is Blonsky and Probst (2008),
who deal with a two-player game with incomplete information regarding the discount factors.
Their paper characterizes the efficient equilibria and the Pareto frontier of the payoffs in a game that is related to the Prisoner's Dilemma,
but with more strategies in the stage game, which correspond to different levels of cooperation or trust.
Our paper differs from Blonsky and Probst (2008) in several aspects.
First, we deals with a different class of games that do not include a mechanism of gradually building trust embedded.
Second, while Blonsky and Probst (2008) takes a simple information structure and ask what are the effective equilibria,
we take a simple and natural strategy (grim trigger) and check checks when it induces an equilibrium
for a given information structure.
In other words, while Blonsky and Probst (2008) show how an equilibrium with cooperation look like in a specific game and information structure,
we study the requirements on the players' beliefs that guarantee the existence of an equilibrium in grim trigger strategies.

Other papers that deal with cooperation in repeated games with incomplete information,
albeit in settings quite different from ours, are Watson (1999,2002), who analyze the building of trust in a two-player game
with incomplete information regarding the payoff matrix,
somewhat similar to the game in Blonsky and Probst (2008),
and Kajii and Morris (1997) and Chassang and Takahashi (2011),
who investigate the robustness of equilibria in repeated games to a small amount of incomplete information on the payoff matrix.
Chassang and Takahashi (2011) also deal specifically with the Prisoner's Dilemma and show
that the grim trigger equilibria are not the most robust way to sustain cooperation.

The paper is organized as follows.
In Section \ref{section model} we present the repeated Prisoner's Dilemma with
incomplete information on the discount factor, and the concept of cooperation events.
In Section \ref{section characterization} we characterize pairs of cooperation events.
In Section \ref{beliefsec} we present the concepts of $f$-belief, common $f$-belief and iterated $f$-belief,
and characterize the largest pair of cooperation events in the repeated Prisoner's Dilemma using these notions.
Examples are provided in Section \ref{section examples}.
In Section \ref{section generalizations} we present
several extensions and additional results,
including general two-player games with incomplete information on the other player's discount factor,
 the analysis of almost complete information,
 and another possible application of the concept of $f$-belief,
 to define sufficient conditions for ``combining'' two different complete information equilibria into an incomplete information equilibrium,
 not necessarily in a repeated game.

\section{The Model}
\label{section model}

Consider the repeated Prisoner's Dilemma with incomplete
information on the discount factors. That is, the set of states of
nature is $S = [0,1)^2$; a state of nature $s =
(\lambda_1,\lambda_2) \in S$ indicates the discount factors
$\lambda_1$ and $\lambda_2$ of the two players 1 and 2
respectively. An arbitrary player will be denoted by $i$; when $i$
is a player, $j$ denotes the other player.

We consider a general (that is, not necessarily finite) set of states of the world $\Omega$,
equipped with a $\sigma$-algebra $\Sigma$.
The function $\lambda  = (\lambda_1,\lambda_2) : \Omega \to S$
indicates the state of nature that corresponds to each state of the world;
thus, $\lambda_i(\omega) \in [0,1)$ is the discount factor of player $i$ at the state of the world $\omega$.

The belief of each player $i$ is given by a function $P_i : \Omega \to \Delta(\Omega)$.
We denote by $P_i(E\mid\omega)$ the probability that player $i$ ascribes to the event $E \in \Sigma$
at the state of the world $\omega$, and by $\E_i(\cdot\mid\omega)$ the corresponding expectation operator.
The function $P_i$ is measurable in the sense that, for every $E \in \Sigma$,
the function $\omega \mapsto P_i(E\mid\omega)$ is measurable.
This function is required to be consistent, in the sense that each player knows his belief:
$P_i(\{\omega':P_i(\omega)=P_i(\omega')\}\mid\omega)=1$, for every $\omega\in\Omega$.
We assume throughout%
\footnote{Similar though slightly weaker results to the ones we obtain below can be proven when the players have incomplete
information regarding their own discount factor as well as the other player's discount factor.
However, the assumption that a player knows his own discount factor is natural in many applications,
and it allows us to concentrate on the impact of incomplete information on the other player's discount factor.}
that each player $i$ knows his own discount factor in every state of the world $\omega$:
$P_i(\{\omega':\lambda_i(\omega')=\lambda_i(\omega)\}\mid\omega)=1$ for every $\omega\in\Omega$.
The triplet $(\Omega,\Sigma,(P_1,P_2))$ is called a {\em belief space}.

The {\em type} of player $i$ at the state of the world $\omega$
is the set of all states of the world that are indistinguishable
from the true state of the world, given the player's information;
formally, it is the set
$\{\omega'\mid P_i(\omega')=P_i(\omega)\}$.

To allow general information structures, we assume that the information that a player has
is described by a sub-$\sigma$-algebra $\Sigma_i$ of $\Sigma$,
and that both $\lambda_i$ and $P_i$ are $\Sigma_i$-measurable.
Thus, each player knows his discount factor and his beliefs.
The $\sigma$-algebra $\Sigma_i$ identifies the events that can be described by player $i$.

The game is played as follows:
at the outset of the game,
a state of the world $\omega \in \Omega$ is realized, each player $i$ learns his type, and therefore also his discount factor $\lambda_i(\omega)$.
Then the players repeatedly play the Prisoner's Dilemma that appears in Figure \ref{figure1} (see page \pageref{figure1}).

A \emph{course of action} of player $i$ is a function that assigns a
mixed action of player $i$ to each finite history of actions in
the repeated Prisoner's Dilemma.
This is a strategy for the player given his type.

Two simple courses of action of player $i$ are the {\em always defect} course of action $D^*_i$,
in which he always plays $D$,
and the {\em grim trigger} course of action $GT^*_i$,
in which he cooperates at the first stage, and in every subsequent stage he cooperates
only if the other player cooperated in all previous stages.

A \emph{strategy} of player $i$ is a $\Sigma_i$-measurable%
\footnote{So that $\Sigma_i$-measurability of strategies is well defined,
one needs to introduce a topological structure on the space of courses of actions.
In this paper we will study only a simple type of strategies, conditional grim trigger strategies,
which is defined below (Definition \ref{grimdef}),
and the measurability issue will not arise.}
function $\eta_i$ that assigns a course of action $\eta_i(\omega)$
to each state of the world $\omega$.

The {\em subjective payoff} of player $i$ under
the strategy profile $(\eta_i,\eta_j)$, conditional on the state
of the world $\omega$, is
$\gamma_i(\eta_i,\eta_j\mid\omega)=\E_i(\sum_{t=1}^\infty(\lambda_i(\omega))^{t-1}u_{i,t}\mid\omega)$,
where $u_{i,t}$ is player $i$'s payoff at stage $t$.
Note that the subjective payoff depends on $\omega$ in three ways:
first, player $i$'s discount factor depends on $\omega$;
second, the course of action that player $i$ takes depends on $\omega$;
and third, the belief of player $i$ about player $j$'s type, and therefore about player $j$'s course of action,
also depends on $\omega$.

In the complete information case both players know each other's discount factor.
In this case, the grim trigger course of action is an equilibrium if and only if the discount factors of both players
are at least $\lambda_i^0:={1\over3}$.
Thus, the players will cooperate in some states of the world,
and will never cooperate in the remaining states of the world.
This observation leads us to the following definition of
a {\em conditional grim trigger} strategy,
in which the player plays a grim trigger course of action in some states of the world,
and always defects in the rest of the states of the world.
\begin{definition}
\label{grimdef}
Let $i \in \{1,2\}$ be a player, and let $K_i\subseteq\Omega$ be a
$\Sigma_i$-measurable set.
The {\em conditional grim trigger strategy with cooperation region $K_i$} for player $i$,
denoted by $\eta_i^*(K_i)$, is the strategy defined by
$$
\eta^*_i(K_i,\omega)=
\left\{\begin{array}{lll}
GT^*_i&\ \ \ \ \ & \hbox{if }\omega\in K_i,\\
D^*_i& & \hbox{if }\omega\notin K_i.
\end{array}
\right.
$$
\end{definition}

If a player plays a conditional grim trigger strategy,
his action at the first stage reveals whether he follows
the grim trigger course of action or the ``never cooperate'' course of action.

If $\eta^*(K_1,K_2) := (\eta^*_1(K_1),\eta^*_2(K_2))$ is a Bayesian equilibrium,
then whenever $\omega \in K_1\cap K_2$ the players will cooperate all along the game.

\begin{definition}
\label{defin cooperation}
The pair of events $(K_1,K_2)$ is called a pair of \emph{cooperation events}
if $\eta^*(K_1,K_2) = (\eta^*_1(K_1),\eta^*_2(K_2))$ is a Bayesian equilibrium,
\end{definition}

As mentioned in the introduction, the significance of cooperation events is that they enable cooperation
all through the game, without the need to exchange information.
As Example \ref{example new} below shows, there are equilibria in which the players cooperate
from some stage on, after a short signalling period,
even when there are no non-trivial cooperation events.

The pair $(\emptyset,\emptyset)$ is a pair of cooperation events in which the players never cooperate.
Moreover, if $(K_1,K_2)$ are cooperation events and $K_i=\emptyset$, then $K_j=\emptyset$.
Indeed, if $K_i=\emptyset$ then player $i$ never cooperates,
so that player $j$'s dominant strategy is never to cooperate.

Note that there may be two pairs of non-empty yet disjoint cooperation events.
Indeed, given a belief space with a pair of non-empty cooperation events $(K_1,K_2)$,
consider an auxiliary belief space that includes two independent copies of the original belief space.
Then the two copies of $(K_1,K_2)$ are two pairs of disjoint cooperation events in the auxiliary belief space.

The larger the sets $(K_1,K_2)$, the higher the probability that cooperation will occur.
If there are several pairs of sets $(K_1,K_2)$ for which $\eta^*(K_1,K_2)$ is a Bayesian equilibrium,
then using pre-play communication or a mediator the players can agree on the Bayesian equilibrium
with largest pair of cooperation events $(K_1,K_2)$.
Our main goal is to characterize the cooperation events,
and to provide an algorithm for calculating the largest pair of cooperation events $(K_1,K_2)$.

\section{Characterization of the Cooperation Events}
\label{section characterization}

We start by characterizing cooperation events in the (simple) benchmark case of complete information,
where the discount factors are common knowledge among the players.
We will then analyze the general case with incomplete information.

When the game has complete information,
in equilibrium each player knows the course of action that the other player will take:
a player cooperates if and only if his discount factor is at least $\lambda_i^0=\frac{1}{3}$
and the other player cooperates.
In particular, if $\eta^*(K_1,K_2)$ is a Bayesian equilibrium then necessarily $K_1=K_2$:
We therefore obtain the following characterization for cooperation events in the complete information case.

\begin{theorem}
\label{theorem complete}
When the game has complete information, $(K_1,K_2)$ is a pair of cooperation events
if and only if $K_1=K_2\subseteq\Lambda$, where
$\Lambda:=\Lambda_1\cap\Lambda_2$
and $\Lambda_i := \{\omega \in \Omega \colon \lambda_i(\omega)\ge\lambda_i^0\}$ for $i \in \{1,2\}$.
\end{theorem}

When information is incomplete
the characterization of cooperation events is more subtle,
and is described by the following theorem.

\begin{theorem}
\label{conditionprop}
Let $K_i\subseteq\Omega$ be a $\Sigma_i$-measurable event for each $i=1,2$.
The pair $(K_1,K_2)$ is a pair of cooperation events if and only if, for each $i=1,2$
\begin{enumerate}
\item[(a)] $K_i\subseteq\Lambda_i$;
\item[(b)] $P_i(K_j\mid\omega)\ge f_i(\omega)$ for every $\omega\in K_i$; and
\item[(c)] $P_i(K_j\mid\omega)\le f_i(\omega)$ for every $\omega\notin K_i$;
\end{enumerate}
where $f_i$ is the $\Sigma_i$-measurable function defined by $f_i(\omega):=\frac{1-\lambda_i(\omega)}{2\lambda_i(\omega)}$.
\end{theorem}

Condition (a) is an individual rationality condition:
a player will not cooperate if his own discount factor is not sufficiently high to justify cooperation.
Conditions (b) and (c) capture the whole hierarchy of beliefs that should hold to ensure cooperation.
Condition (b) requires that whenever a player is supposed to cooperate,
he assigns sufficiently high probability that the other player will cooperate;
if this were not the case, the player would find it beneficial not to cooperate at that state of the world.
The cutoff for cooperation depends on the player's discount factor, and therefore it depends on the state of the world.
Indeed, if the player's discount factor is high,
he will be willing to take the risk and cooperate in the first stage even if the probability of the event
that the other player will also cooperate is low,
because the possible loss in the first stage is small due to the high discount factor.
Condition (c) is analogous to condition (b); it requires that whenever a player is not supposed to cooperate,
he assigns low probability to the event that the other player will cooperate;
otherwise, the player would prefer to cooperate.
Note that condition (c) trivially holds for $\omega \notin \Lambda_i$ because in this case $f_i(\omega) \geq 1$.

The proof of the theorem
appears in Appendix \ref{appendix proof}.

It is interesting to note that Theorem \ref{theorem complete} is a special case of Theorem \ref{conditionprop}.
When information is complete, whether or not the state of the world $\omega$ lies in $K_1$ or $K_2$ is
common knowledge among the players, and therefore $P_i(K_j\mid\omega)$ is either 0 or 1.
If $K_1=K_2$ then the probability in condition (b) is 1 and the probability in condition (c) is 0.
Since $K_i\subseteq\Lambda_i$, we have $f_i(\omega)\le1$ for $\omega\in K_i$. Moreover, $f_i>0$ for every $\omega \in \Omega$.
It follows that in this case conditions (b) and (c) hold.
On the other hand, if $K_1 \neq K_2$ and $\omega\in K_i\setminus K_j$, then $P_i(K_j\mid\omega)=0$, and condition (b) does not hold.
Therefore, as Theorem \ref{theorem complete} states, we deduce that the conditions that ensure that $(K_1,K_2)$ are cooperation events
are $K_1=K_2\subseteq\Lambda_1\cap\Lambda_2$.

\section{Belief Operators}
\label{beliefsec}

In this section we present the concepts of $f$-belief and
common $f$-belief,
which are generalizations of the concepts of $p$-belief and
common-$p$-belief introduced by Monderer and Samet (1989).
These concepts will be useful in the construction of cooperation events.

The definitions and results presented in this section
are valid for general belief spaces with any finite number of players
(see Definition 10.1 in Maschler, Solan, and Zamir (2013) for
a formal definition of a general belief space).

\begin{definition}
\label{definition fbelief}
Let $i$ be a player, and let
$f_i:\Omega \rightarrow \mathbb{R}$ be a $\Sigma_i$-measurable function, let
$A\subseteq\Omega$ be an event, and let $\omega \in A$.
Player $i$ \emph{$f_i$-believes} in the event $A$ at the state of the world $\omega$ if $P_i(A\mid\omega)\ge f_i(\omega)$.
\end{definition}

We say that both players \emph{$f$-believe in an event at the state of the world $\omega$} if player 1 $f_1$-believes in the event at $\omega$
and player 2 $f_2$-believes in the event at $\omega$.
A function $f:\Omega \rightarrow \mathbb{R}^2$ is called \emph{*-measurable} if $f_i$ is $\Sigma_i$-measurable for each $i \in \{1,2\}$.

\begin{definition}
\label{definition fbelief1}
Let $f:\Omega \rightarrow \mathbb{R}^2$ be a *-measurable function,
let
$A\subseteq\Omega$ be an event, and let $\omega \in A$.
The event $A$ is \emph{common $f$-belief} at the state of the world $\omega$ if at
$\omega$ each player $i$ $f_i$-believes in $A$,
each player $i$ $f_i$-believes that both players $f$-believe in $A$,
each player $i$ $f_{i}$-believes that both players $f$-believe that both players $f$-believe in $A$,
etc. ad infinitum.
\end{definition}

If there is $p \in (0,1)$ such that $f_1\equiv p$ and $f_2\equiv p$,
then $f$-belief and common $f$-belief reduce to Monderer and Samet's (1989) $p$-belief
and common $p$-belief, respectively.
If $f_i$ is a constant function $p_i$ for each $i \in \{1,2\}$,
but $p_1$ and $p_2$ may differ, then the concept of $f$-belief reduces to the concept of
$\textbf{p}$-belief of  Morris and Kajii (1997), where $\textbf{p} = (p_1,p_2)$.

An equivalent formulation of conditions (b) and (c) of Theorem
\ref{conditionprop} using the concept of $f_i$-belief, where $(f_i)_{i=1,2}$ are the functions that are defined in the statement of the theorem,
is the following:
\begin{enumerate}
  \item[(b')] $K_j$ is an $f_i$-belief in $K_i$; and
  \item[(c')] $K_j^c$ is a $(1-f_i)$-belief in $K_i^c$.
\end{enumerate}
Thus, the events $(K_1,K_2)$ are cooperation events if
(a) each $K_i$ is a subset of $\Lambda_i$ (individual rationality condition),
(b') at each state of the world in which player $i$ is supposed to cooperate,
he $f_i$-believes that the other  player  is going to cooperate,
and (c') at each state of the world in which player $i$ is not supposed to cooperate,
he $(1-f_i)$-believes that the other  player  is not going to cooperate.

\begin{definition}
Let $i$ be a player, and let
$f_i:\Omega \rightarrow \mathbb{R}$ be a $\Sigma_i$-measurable function.
The \emph{$f_i$-belief operator} of player $i$ is the operator
$B_i^{f_i}:\Sigma \rightarrow \Sigma$ that assigns to each event the
states of the world at which player $i$ $f_i$-believes in the event:
$B_i^{f_i}(A):=\{\omega\in\Omega \mid P_i(A\mid\omega)\ge
f_i(\omega)\}$.
\end{definition}

Because both $f_i$ and $P_i$ are $\Sigma_i$-measurable functions,
the set $B_i^{f_i}(A)$ is in $\Sigma_i$ for every event $A$.

Let $f:\Omega \rightarrow \mathbb{R}^2$ be a *-measurable function.
Using the $f_i$-belief operator we can formally define common $f$-belief.
Define
\begin{eqnarray}
D^{0,f}(C)&:=&C,\\
D^{n+1,f}(C)&:=&B_1^{f_1}(D^{n,f}(C)) \cap B_2^{f_2}(D^{n,f}(C)) \hbox{ for every } n\ge0,\\
D^f(C)&:=&\bigcap_{n\ge 1} D^n(C).
\end{eqnarray}
An event $C$ is \emph{common $f$-belief} at $\omega$ if and only if
$\omega\in D^f(C)$.

The following proposition, which lists several desirable properties that the $f$-belief operator satisfies,
is analogous to results derived in Monderer and Samet (1989).

\begin{proposition}
\label{prop1}
Let $i \in \{1,2\}$ and $f : \Omega \to \dR^2$ be *-measurable. The following statements hold:
\begin{enumerate}
  \item If $A,C\in\Sigma$ and $A\subseteq C$, then $B_i^{f_i}(A)\subseteq B_i^{f_i}(C)$.
  \item If $A\in\Sigma$ then $B_i^{f_i}(B_i^{f_i}(A))=B_i^{f_i}(A)$.
  \item If $(A_n)_{n=1}^\infty$ is a decreasing sequence of events, then
  $B_i^{f_i}(\bigcap_{n=1}^\infty A_n)=\bigcap_{n=1}^\infty B_i^{f_i}(A_n)$.
  \item If $C\in\Sigma$ is a $\Sigma_i$-measurable event,
  then
  \[ B_i^{f_i}(C)=(C\setminus\{\omega\in\Omega\mid f_i(\omega)>1\})\cup\{\omega\in\Omega\mid f_i(\omega)\le0\}. \]
  \item If $f_i>0$ or $\{f_i\le0\}\subseteq C$, then
        $B_i^{f_i}(A)\cap C=B_i^{f_i}(A\cap C)$ for every event $A\in\Sigma$ and every $\Sigma_i$-measurable event $C\in\Sigma$.
\end{enumerate}
\end{proposition}

\begin{proof}
The proof of parts 1 and 3 is similar to
the proof for Proposition 2 in Monderer and Samet (for $p$-belief
operators). To prove part 4, observe that $P_i(C\mid\omega)=1$ for every $\omega\in C$,
and $P_i(C\mid\omega)=0$ for every $\omega\notin C$. Part 2 follows from part 4, because for
every $A\in\Sigma$, $B_i^{f_i}(A)$ is a $\Sigma_i$-measurable event that contains
$\{\omega\in\Omega\mid f_i(\omega)\le0\}$. To prove part 5, assume
$\{f_i\le0\}\subseteq C$. In this case
$B_i^{f_i}(C)=C\setminus\{f_i>1\}$. Assume $\omega\in B_i^{f_i}(A\cap C)$.
From part 1 we have $\omega\in B_i^{f_i}(A)$ and $\omega\in
B_i^{f_i}(C)\subseteq C$. For the other direction, assume
$\omega\in B_i^{f_i}(A)\cap C$. Then $P_i(A\mid\omega)\ge f_i(\omega)$
and $P_i(C\mid\omega)=1$, and therefore $P_i(A\cap C\mid\omega)\ge
f_i(\omega)$; that is, $\omega\in B_i^{f_i}(A\cap C)$.
\end{proof}

\bigskip

By Proposition \ref{prop1}, $B_i^{f_i}$ is a belief operator as defined by Monderer and Samet (1989).
Analogously to Proposition 3 in Monderer and Samet (1989) we obtain the following characterization of common $f$-belief.

\begin{proposition}
\label{prop2}
The event $C$ is a common $f$-belief at the state of the world $\omega$
if and only if
there exists an event $D\in\Sigma$ such that (a) $\omega\in D$, (b)
$D\subseteq B_i^{f_i}(C)$, and (c) $D\subseteq B_i^{f_i}(D)$ for every $i\in
\{1,2\}$.
\end{proposition}

\subsection{Common $f$-Belief and Cooperation Events}
\label{fbeliefsec}

In the spirit of the notion of iterated $p$-belief (see Morris, 1999),
we define the concept of iterated $f$-belief of pairs of events.
As we will see below, this notion is also related to the concept of common $f$-belief.

\begin{definition}
For every two events $C_i\in\Sigma$, $i=1,2$, define
\begin{eqnarray*}
D_i^{1,f}(C_i,C_j)&:=&B_i^{f_i}(C_j)\cap C_i,\\
D_i^{n,f}(C_i,C_j)&:=&B_i^{f_i}(D_j^{n-1,f}(C_j,C_i))\cap D_i^{n-1,f}(C_i,C_j), \hbox{ for every } n > 1,\\
D_i^f(C_i,C_j)&:=&\bigcap_{n\ge 1}D_i^{n,f}(C_i,C_j).
\end{eqnarray*}
The event $D_i^f(C_1,C_2)$ is called the {\em iterated $f$-belief of player $i$ w.r.t. $(C_1,C_2)$}.
\end{definition}

To understand this definition,
suppose that $\omega$ is the true state of the world, and $C_1$ and $C_2$ are two events that contain
$\omega$.
Suppose that each player $i$ is informed of $C_i$.
The statement $\omega \in D_i^f(C_i,C_j)$ is then equivalent to the following:
(a) player $i$ $f_i$-believes that $\omega$ is in player $j$'s set,
(b) player $i$ $f_i$-believes that (b1) $\omega$ is in player $j$'s set and that
(b2) player $j$ $f_j$-believes that $\omega$ is in player $i$'s set,
etc. ad infinitum.

Because $B^f_i(A)$ is a $\Sigma_i$-measurable event for every event $A$,
if $C_i$ is a $\Sigma_i$-measurable event, then the event $D_i^f(C_1,C_2)$, the iterated $f$-belief of player $i$ w.r.t.~$(C_1,C_2)$,
is also $\Sigma_i$-measurable, for every (not necessarily $\Sigma_i$-measurable) event $C_j$.

As the next lemma states, $D_1^f(C_1,C_2)$ and $D_2^f(C_2,C_1)$
are the largest subsets of $C_1$ and $C_2$, respectively, such
that condition (b) in Theorem \ref{conditionprop} holds.

\begin{lemma}
\label{lemma1}
\begin{enumerate}
  \item For each $i\in\{1,2\}$ and every $\omega\in D_i^f(C_i,C_j)$ one has $P_i(D_j^f(C_j,C_i)\mid\omega)\ge f_i(\omega)$.
        $D_1^f(C_1,C_2)$ and $D_2^f(C_2,C_1)$ are the largest subsets of $C_1$ and $C_2$ (respectively) such that
        this property holds.\footnote{They are the largest in the following (strong) sense: if $K_1\subseteq C_1$ and
        $K_2\subseteq C_2$ satisfy $P_i(K_j\mid\omega)\ge f_i(\omega)$ for each $i=1,2$ and every $\omega\in K_i$, then
        $K_i\subseteq D_i^f(C_i,C_j)$.}
  \item If $C_i$ is $\Sigma_i$-measurable and $\{f_i\le0\}\subseteq C_i$ for $i=1,2$,
  then $D_i^{n,f}(C_i,C_j)$ and $D_i^f(C_i,C_j)$ depend only on the intersection of $C_i$ and $C_j$.
  In this case, $D_1^f(C_1,C_2)\cap D_2^f(C_2,C_1)=D^f(C_1\cap C_2)$,
  the event containing all states of the world $\omega$ such that $C_1\cap C_2$
  is a common $f$-belief at $\omega$, and $D_i^f(C_i,C_j)=B_i^{f_i}(D^f(C_1\cap C_2))$,
  the event containing all states of the world $\omega$ where
  player $i$ $f$-believes that $C_1\cap C_2$ is a common $f$-belief.
\end{enumerate}
\end{lemma}

Lemma \ref{lemma1}(2) relates the concepts of common $f$-belief and iterated $f$-belief:
when $f_i(\omega) > 0$ for every $\omega \not\in C_i$ and each $i=1,2$,
the event $D^f(C_1\cap C_2)$
is the intersection of the iterated $f$-belief events $D_1^f(C_1,C_2)$ and $D_2^f(C_1,C_2)$.
Example 3 in Morris (1999) shows that the concept of iterated belief that we defined here is different from
the concept with the same name defined in Morris (1999).

If $f_1>0$ and $f_2>0$,
then the set $\{f_i \leq 0\}$ is empty,
and part 2 holds as soon as $C_i$ is $\Sigma_i$-measurable.
The condition $f_1>0$ and $f_2>0$ holds in particular for the case of $p$-beliefs, when $f_1=f_2\equiv p$ for some $p \in (0,1]$.

\bigskip

\begin{proof}
To avoid cumbersome notation, we write $D_i^f$ and $D_i^{n,f}$ instead of $D_i^f(C_i,C_j)$ and $D_i^{n,f}(C_i,C_j)$, respectively.

We first argue that $D_i^f\subseteq B_i^{f_i}(D_j^f)$. Indeed,
by Proposition \ref{prop1}(3), and since $(D_j^{n,f})_{n=1}^\infty$ is a decreasing sequence of events,
        $$
        B_i^{f_i}(D_j^f)=B_i^{f_i}\left(\bigcap_{n\ge 1}D_j^{n,f}\right)=\bigcap_{n\ge 1}B_i^{f_i}(D_j^{n,f})\supset
        \bigcap_{n\ge 1} \left(B_i^{f_i}(D_j^{n,f})\cap D_i^{n,f}\right)=
        $$
        $$
        =\bigcap_{n\ge 1}D_i^{n+1,f}=D_i^f.
        $$
For the maximality property, assume that $K_1\subseteq C_1$ and
$K_2\subseteq C_2$ satisfy $P_i(K_j\mid\omega)\ge f_i(\omega)$ for every $\omega\in K_i$ and each $i=1,2$.
For each $i \in \{1,2\}$ we then have
$K_i\subseteq B_i^{f_i}(K_j)$,
which implies that $D_i^{1,f}(K_i,K_j)=K_i$.
It follows by induction on $n$ that $D_i^{n,f}(K_i,K_j)=K_i$,
and therefore $D_i^f(K_i,K_j)=K_i$.
Because $K_i\subseteq C_i$ it follows from Proposition \ref{prop1}(1)
that $D_i^f(K_i,K_j)\subseteq D_i^f(C_i,C_j)$, and the first claim follows.

Denote $C:=C_1\cap C_2$.
Because $\{f_i\le0\}\subseteq C_i$, by Proposition \ref{prop1}(5) one has $D_i^{1,f}(C_i,C_j)=B_i^{f_i}(C)$.
        This proves the first part of the second claim.

Because $D_i^{1,f}(C_i,C_j)=B_i^{f_i}(C)$, it can be verified
 from the definition that for $n\ge1$ one has
        $D^{n,f}(C)=D_1^{n,f}(C_1,C_2)\cap D_2^{n,f}(C_2,C_1)$, and therefore $D^f(C)=D_1^f(C_1,C_2)\cap D_2^f(C_2,C_1)$.
        The event ``player $i$ $f_i$-believes that $C_1\cap C_2$ is a common $f$-belief'' is the event
\begin{eqnarray}
B_i^{f_i}(D^f(C))&=&B_i^{f_i}(D_1^f(C_1,C_2)\cap D_2^f(C_2,C_1))\\
        &=&B_i^{f_i}(D_j(C_j,C_i))\cap D_i^f(C_i,C_j)=D_i^f(C_i,C_j).
\end{eqnarray}
Since $D_i^f(C_1,C_2)$ is $\Sigma_i$-measurable,
the second equality follows from Proposition \ref{prop1}(5),
        and the last one from the first part of this lemma.
\end{proof}

\bigskip

As a conclusion of Theorem \ref{conditionprop} and Lemma \ref{lemma1} we obtain the following result.

\begin{theorem}
\label{2actioncor}
If $C_i\subseteq\Lambda_i$ for $i=1,2$,
then $(B_1^{f_1}(D^f(C_1\cap
C_2)),B_2^{f_2}(D^f(C_1\cap C_2)))$ is a pair of cooperation events if and only if $P_i(B_j^{f_j}(D^f(C_1\cap
C_2))\mid\omega)\le f_i(\omega)$ for every $\omega\in
\Lambda_i\setminus B_i^{f_i}(D^f(C_1\cap C_2))$, for $i=1,2$.
In particular, $(B_1^{f_1}(D^f(\Lambda)),B_2^{f_2}(D^f(\Lambda)))$ is a pair of cooperation events, where $\Lambda=\Lambda_1\cap\Lambda_2$.
\end{theorem}

\begin{proof}
The first claim follows from Theorem \ref{conditionprop} and Lemma \ref{lemma1},
  because $f_i(\omega)>1$ for every $\omega\not\in\Lambda_i$.

To prove the second claim, assume to the contrary that there exists $\omega^*\in \Lambda_1\setminus B_1^{f_1}(D^f(\Lambda))$
  such that $P_1(B_2^{f_2}(D^f(\Lambda)\mid\omega^*))> f_1(\omega^*)$.
  Denote $K_1:=B_1^{f_1}(D^f(\Lambda)\cup\{\omega^*\})$ and $K_2:=B_2^{f_2}(D^f(\Lambda))$.
  From the assumption, $P_i(K_j\mid\omega)\ge f_i(\omega)$ for every $\omega\in K_i$ and each $i\in\{1,2\}$,
contradicting Lemma \ref{lemma1}(1).
\end{proof}

\bigskip

The last case defines the largest cooperation events: the profile
in which each player $i$ plays the grim trigger course of action
whenever he $f_i$-believes that it is a common $f$-belief that
$\lambda_1,\lambda_2\ge\frac{1}{3}$ (and otherwise, always
defects) is an equilibrium.

Note that although there cannot be a pair of cooperation events larger than
$(B_1^f(D^f(\Lambda)),B_2^f(D^f(\Lambda)))$, there may be smaller pairs of non-trivial cooperation events (see Example \ref{prisonerex1} below).

\begin{remark}
\label{ICR}
If we consider the solution concept of interim correlated rationalizability (ICR) instead of Bayesian equilibrium, we obtain slightly different results.
First, the course of action $D^*_i$ is always rationalizable for player $i$, regardless of the information structure,
which makes condition (c) in Theorem \ref{conditionprop} redundant.
Second, using similar arguments as in the proofs of Theorems \ref{conditionprop} and \ref{2actioncor}, it can be shown that $GT_i^*$ is rationalizable for player $i$ if $\omega\in B_i^{f_i}(D^f(C_1\cap C_2))$ for some $C_i\subset \Lambda_i$.
Therefore, the strategy $\eta^*_i(B_i^{f_i}(D^f(C_1\cap C_2)))$ is rationalizable for \emph{every} choice of $C_i\subset \Lambda_i$.
This difference is more significant when considering larger games (see Section \ref{g2pg}).
\end{remark}

\section{Examples}
\label{section examples}

In this section we present several examples that illustrate properties of cooperation events.
Constructing examples where the only pair of cooperation events is the trivial pair $(\emptyset,\emptyset)$ is not difficult.
For example, if there is a player $i$ who, whenever his discount factor is higher than $\lambda_i^0 = \frac{1}{3}$,
assigns high probability to the event that player $j$'s discount factor is low, then that player will never cooperate,
and there will be no non-trivial pair of cooperation events.
The same phenomenon will occur if such a property holds in higher levels of belief,
e.g.,
there is a player $i$ who believes that player $j$'s belief satisfies that,
whenever player $j$'s discount factor is higher than $\lambda_j^0$,
player $j$ assigns high probability to the event that player $i$'s discount factor is low.

We start with two examples in which $B_i^{f_i}(D^f(\Lambda)) = \Lambda_i$ for $i \in \{1,2\}$,
so that $(\Lambda_1,\Lambda_2)$ is a pair of cooperation events;
in the first example, there are additional pairs of non-trivial cooperation events,
while in the second example $(\Lambda_1,\Lambda_2)$ is the unique pair of non-trivial cooperation events.

In all the examples in this section, $\Omega\subseteq[0,1)^2$, and
we interpret the coordinates of $\omega\in\Omega$ as the players'
discount factors, i.e.
$\lambda(\omega)=(\lambda_1(\omega),\lambda_2(\omega))=\omega$. In
all examples, each player's belief
regarding the other player's discount factor depends only on his
own discount factor.

\begin{example}
\label{prisonerex1}
Let $\Omega=\{{1\over4},{1\over2},{3\over4}\}^2$.
The belief of the players is derived from a common prior, which is the uniform distribution over $\Omega$.

We first show that $D_i^{f}(\Lambda_1,\Lambda_2)=\Lambda_i$ for each $i \in \{1,2\}$,
so that by Theorem \ref{2actioncor}, $(\{{1\over2},{3\over4}\},\{{1\over2},{3\over4}\})$ is a pair of cooperation events.
We then argue that $(\{{3\over4}\},\{{3\over4}\})$ is also a pair of cooperation events.%
\footnote{To avoid cumbersome notation, in the examples we write
$\{{3\over4}\}$ instead of $\{{3\over4}\} \times \{{1\over4},{1\over2},{3\over4}\}$ or $\{{1\over4},{1\over2},{3\over4}\} \times \{{3\over4}\}$,
and so on, when the meaning is clear. We also write
$f_1(\frac{3}{4})$ and $f_2(\frac{3}{4})$ instead of $f_1((\frac{3}{4},y))$ and $f_2((y,\frac{3}{4}))$
for $y \in \Omega$.}
Note that
\[ f_i\left(\frac{3}{4}\right) = \frac{1}{6}, \ \ \ f_i\left(\frac{1}{2}\right) = \frac{1}{2}, \ \ \ f_i\left(\frac{1}{4}\right) = \frac{3}{2}. \]

Because $\lambda_1^0=\lambda_2^0={1\over3}$, it follows that
$\Lambda_1=\{{1\over2},{3\over4}\}\times\{{1\over4},{1\over2},{3\over4}\}$ and
$\Lambda_2=\{{1\over4},{1\over2},{3\over4}\}\times\{{1\over2},{3\over4}\}$.
Therefore, $P_1(\Lambda_2\mid\omega)=\frac{2}{3}$ for every $\omega$. Since
$f_1({3\over4})<\frac{2}{3}$ and $f_1({1\over2})<\frac{2}{3}$ we get
$D_1^{1,f}(\Lambda_1,\Lambda_2)=\Lambda_1$. Similarly,
$D_2^{1,f}(\Lambda_2,\Lambda_1)=\Lambda_2$.
It follows that
$B_1^{f_1}(D^f(\Lambda))=D_1^f(\Lambda_1,\Lambda_2)=\Lambda_1$, and therefore
for each $i \in \{1,2\}$, the iterated $f$-belief of player $i$ w.r.t. $(\Lambda_1,\Lambda_2)$ is $\Lambda_i$.

To see that $(\{{3\over4}\},\{{3\over4}\})$ is a pair of cooperation events,
note that
\[ P_1\left(\left\{{1\over4},{1\over2},{3\over4}\right\}\times \left\{\frac{3}{4}\right\} \mid \omega\right) =
P_2\left(\left\{{3\over4}\right\}\times \left\{{1\over4},{1\over2},\frac{3}{4}\right\} \mid \omega\right) = \frac{1}{3}, \]
for every state of the world  $\omega \in \Omega$ and each player $i \in \{1,2\}$,
and use Theorem \ref{conditionprop}.
In fact, the same argument also implies that $(\{{1\over2},{3\over4}\},\{{1\over2},{3\over4}\})$ is a pair of cooperation events,
because $P_i(\{{1\over2},{3\over4}\} \mid\omega) = \frac{2}{3}$ for every state of the world $\omega \in \Omega$ and each player $i \in \{1,2\}$.
\end{example}

\begin{example}
\label{prisonerex2}

Let $\Omega$ be as in the Example \ref{prisonerex1}.
The beliefs of the players are given in Figure 2.

\[ \begin{array}{c||c|c|c|}
 & \lambda_j={1\over4} & \lambda_j={1\over2} & \lambda_j={3\over4}\\
\hline\hline
\lambda_i={1\over4} & 1\over3 & 1\over3  & 1\over3\\
\hline
\lambda_i={1\over2} & 0 & 1\over3 & 2\over3\\
\hline
\lambda_i={3\over4} & 0  & 0 & 1 \\
\hline
\end{array}
\]
\centerline{Figure 2: The beliefs of player $i$ given his discount factor in Example \ref{prisonerex2}.}
\bigskip

Each player believes that the other
player's discount factor is at least as high as his own.
One can verify that $B_i^{f_i}(D^f(\Lambda))=\Lambda_i$ for $i \in \{1,2\}$,
so that $(\{{1\over2},{3\over4}\},\{{1\over2},{3\over4}\})$ is a pair of cooperation events.

Let $C_i$ be as in Example \ref{prisonerex1}.
Because
\[ P_1\left(\left\{{1\over4},{1\over2},{3\over4}\right\}\times\left\{\frac{3}{4}\right\} \mid {1\over2}\right) =
P_2\left(\left\{{3\over4}\right\}\times\left\{{1\over4},{1\over2},\frac{3}{4}\right\} \mid {1\over2}\right)=
\frac{2}{3} > \frac{1}{2}, \]
it follows from Theorem \ref{conditionprop}
that $(\{{3\over4}\},\{{3\over4}\})$
is not a pair of cooperation events.
One can verify that in this example, the only non-trivial pair of cooperation events is $(\{{1\over2},{3\over4}\},\{{1\over2},{3\over4}\})$.
\end{example}

In the following example $B_i^{f_i}(D^f(\Lambda))$ is a strict non-empty subset of
$\Lambda_i$ for $i \in \{1,2\}$.

\begin{example}
\label{prisonerex3}
Let $\Omega$ be as in Example \ref{prisonerex1}.
The beliefs of the players are given in Figure 3:
each player believes that his discount factor is high if and only if the other player's discount factor is low.

\[ \begin{array}{c||c|c|c|}
\lambda_i & {3\over4} & {1\over2} & {1\over4}\\
\hline\hline
{3\over4} & 0 & 0 & 1\\
\hline
{1\over2} & 0 & 1 & 0\\
\hline
{1\over4} & 1 & 0 & 0 \\
\hline
\end{array}
\]
\centerline{Figure 3: The beliefs of player $i$ given his discount factor in Example \ref{prisonerex3}.}
\bigskip

One can verify that $B_i^{f_i}(D^f(\Lambda)) = \{{1\over2}\}$,
and that $(\{{1\over2}\},\{{1\over2}\})$ is the unique non-trivial pair of cooperation events.
\end{example}

The next example shows that even if there are no cooperation events,
there may be equilibria in which, with positive probability, the players eventually cooperate.
Such a cooperation is achieved by having players signal their information one to the other in the first stage of the game.

\begin{example}
\label{example new}

Let $\Omega=\{({1\over2},{1\over2}),({1\over2},{1\over4})\}$;
that is, player 1 has a high discount factor,
while the discount factor of player 2 may be high or low.
With probability $p < {1\over2} = f_1({1\over2})$ the state of the world is $({1\over2},{1\over2})$.

We first argue that $B_i^{f_i}(D^f(\Lambda))=\emptyset$ for $i \in \{1,2\}$,
and hence there is no non-trivial pair of cooperation events.
Indeed, $\Lambda_1=\Omega$ and $\Lambda_2=\{({1\over2},{1\over2})\}$,
and therefore $\Lambda := \Lambda_1\cap\Lambda_2 =\{({1\over2},{1\over2})\}$.
Now, $D_1^{1,f}=B_1^f(\Lambda)=B_1^f(\{({1\over2},{1\over2})\})$, and since $p<f_1({1\over2})<f_1({1\over4})$,
we have $D_1^{1,f}(\Lambda)=\emptyset$. Therefore, $B_1^f(D^f(\Lambda))=\emptyset$, and thus $B_2^f(D^f(\Lambda))=\emptyset$.

We claim that the following strategy profile $\sigma=(\sigma_1,\sigma_2)$ is a Bayesian equilibrium.
\begin{itemize}
\item[Player 1:]
Play $D$ in the first stage.
From the second stage on play $D^*_1$ if player 2 played $D$ in the first stage,
and play $GT^*_1$ if player 2 played $C$ in the first stage.
\item[Player 2:]
If the state of the world is $({1\over2},{1\over4})$, play the course of action $D^*_2$.
If the state of the world is $({1\over2},{1\over2})$, play $C$ in the first stage;
from the second stage on play the course of action $GT^*_2$,
starting with action $C$ at the second stage, without taking into account player 1's action at the first stage.
\end{itemize}
Note that under this strategy profile,
the first stage is used by player 2 to signal the state of the world to player 1,
and from the second stage on the players either defect in all stages or cooperate in all stages.

We now verify that this strategy pair is a Bayesian equilibrium.
Because from the second stage on the players follow an equilibrium,
any profitable deviation involves deviating in the first stage.

Player 1 cannot profit by deviating in the first stage, because in the first stage he plays a dominant strategy in the one-shot game,
and his action at that stage does not affect the evolution of the play.

We now argue that player 2 cannot profit by deviating in the first stage either.
Assume first that the state of the world is $({1\over2},{1\over2})$.
Player 2's payoff under $\sigma$ is then ${1\over2}\times\frac{3}{1-{1\over2}} = 3$,
while if he deviates in the first stage and plays $D$
his payoff would be $\frac{1}{1-{1\over2}} = 2$.
Assume now that the state of the world is $({1\over2},{1\over4})$.
Player 2's payoff under $\sigma$ is then $\frac{1}{1-{1\over2}} = 2$,
while if he deviates in the first stage and plays $C$ he will be able to gain 4 in the second stage,
so that his payoff is bounded by ${1\over4}\times 4 + ({1\over4})^2\times\frac{1}{1-{1\over4}} ={13\over12}< 2$.
\end{example}

Our last example concerns the continuity of
cooperation events as the players' beliefs vary. It shows that
even if each player knows approximately the other player's discount factor, there need not be cooperation events that are close to
$(\Lambda_1,\Lambda_2)$. We will discuss this issue further in
Section \ref{section almost}.

\begin{example}
\label{prisonerex4} Let $\Omega=(0,1)^2$ be equipped with the
Borel $\sigma$-algebra. Each player believes that the other
player's discount factor is within $\epsilon>0$ of his own:
player $i$ believes that $\lambda_j$
is uniformly distributed in the interval $(\lambda_i(\omega)-\epsilon,\lambda_i(\omega)+\epsilon) \cap (0,1)$.
One has $\Lambda_1=[{1\over3},1)\times(0,1)$
and $\Lambda_2=(0,1)\times[{1\over3},1)$. We will show that
$B_1^{f_1}(D^f(\Lambda))=[{1\over2},1)\times(0,1)$, provided $\epsilon$ is sufficiently small,\
which will show that $\Lambda_i$ and $B_i^{f_i}(D^f(\Lambda))$ are not close in the Hausdorff metric.

Let $\delta>0$. Because $f_i$ is a continuous
monotonically decreasing function, and because $f_i({1\over2})={1\over2}$,
there exists $\delta'>0$ such that $f_i({1\over2}-\delta')={1\over2}+\delta$.
Therefore, any state of the world $\omega$ such that
$\lambda_i(\omega)<\min\{{1\over3}+2\epsilon\delta,{1\over2}-\delta'\}$ is not
in $D_i^{1,f}(\Lambda_i,\Lambda_j)$.
Indeed, for such
states of the world $\omega$, one has
$P_i(\Lambda_j\mid\omega)=P_i(\{\lambda_j\ge{1\over3}\}\mid\omega)<{1\over2}+\delta=f_i({1\over2}-\delta')<f_i(\omega)$.
Therefore, $D_1^{1,f}(\Lambda_1,\Lambda_2)=[x^1,1)\times(0,1)$ for some
$x^1\ge\min\{{1\over3}+2\epsilon\delta,{1\over2}-\delta'\}$; a similar result holds for player 2.
By the same reasoning, any state of the world $\omega$ such that
$\lambda_i(\omega)<\min\{x^1+2\epsilon\delta,{1\over2}-\delta'\}$ is not
in $D_i^{2,f}(\Lambda_i,\Lambda_j)$, and
$D_1^{2,f}(\Lambda_1,\Lambda_2)=[x^2,1)\times(0,1)$ for
$x^2\ge\min\{x^1+2\epsilon\delta,{1\over2}-\delta'\}$. We continue inductively, and deduce that any state of the world $\omega$ such that
$\lambda_i(\omega)<{1\over2}-\delta'$ is not in
$D_i^f(\Lambda_i,\Lambda_j)$. This is true for every $\delta>0$.
Because $\delta'$ goes to $0$ as $\delta$ goes to $0$,
 $D_1^f(\Lambda_1,\Lambda_2)\subseteq[{1\over2},1)\times(0,1)$.
An analog inequality holds for player 2.

The inclusion $D_1^f(\Lambda_1,\Lambda_2)\supseteq[{1\over2},1)\times(0,1)$ follows from Lemma \ref{lemma1}(1),
because $[{1\over2},1)\times(0,1)$ and $(0,1)\times[{1\over2},1)$ satisfy
inequality (b) of Theorem \ref{conditionprop}.

From Theorem \ref{2actioncor} we deduce that $([{1\over2},1),[{1\over2},1))$ is a pair of cooperation events.
An alternative way to show the last point is to use Theorem \ref{conditionprop}.

Note that this analysis does not change if
$\Omega=(0,1)^2\cap\{(x,y):|x-y|<\epsilon\}$, which verifies that
the true state of the world is within the support of the beliefs
of the players. Also, if, for every state of the world $\omega$,
player $i$ believes that $\lambda_j$ is distributed in any
non-atomic symmetric way around $\lambda_i(\omega)$, the result
still holds. Similar examples can also be constructed with a finite state space.
\end{example}

\section{Generalizations and Additional Results}
\label{section generalizations}

\subsection{General Two-Player Repeated Games}
\label{g2pg}

Our main results were given for the repeated Prisoner's Dilemma.
In this section we provide an analog result for general repeated
games with incomplete information on the discount factors. We will start by
explaining how to adapt Theorem \ref{conditionprop} to this setup.
Conditions (b) and (c) of that theorem will change, because in
general games there are more ways in which a player can deviate
than in the Prisoner's Dilemma. As we will see below, in the
general case, even if the event $D^f(\Lambda_1\cap \Lambda_2)$ is
not empty, there might not be non-trivial cooperation events.
This happens because in games larger than $2 \times 2$,
conditions (b) and (c) do not refer to the same function $f$, but
rather to different functions $f$ and $g$, respectively, and
therefore $D^f(\Lambda_1\cap \Lambda_2)$ does not automatically
fulfill condition (c).
However, in the spirit of Remark \ref{ICR}, if we consider ICR as a solution concept,
the function $g$ becomes irrelevant and $\eta^*_i(B_i^{f_i}(D^f(C_1\cap C_2)))$ are still rationalizable strategies for every choice of $C_i\subset \Lambda_i$,
and in particular for $C_i=\Lambda_i$. Therefore,
unlike the case of Bayesian equilibrium, when considering ICR the results remain essentially the same in larger games.

Consider a two-player repeated game $G$ where the set of actions of each player $i \in \{1,2\}$ is a finite set $A_i$,
and his utility function is $u_i : A_1 \times A_2 \to \dR$.
Each $u_i$ is extended multilinearily to mixed actions.
As in the model of Section \ref{section model},
each player has incomplete information about the other player's discount factor.
Let $\sigma = (\sigma_1,\sigma_2)$ be a Nash equilibrium in mixed actions of
the one-shot game $\Gamma := (\{1,2\}, A_1,A_2,u_1,u_2)$, and let $\tau = (\tau_1,\tau_2)$ be
a pair of pure actions that satisfies the following two conditions, for each $i \in \{1,2\}$:
\begin{itemize}
\item   $u_i(\tau_1,\tau_2) > u_i(\sigma_1,\sigma_2)$, and
\item   $\tau_i$ is not a best response against $\sigma_j$.
\end{itemize}
The actions $\sigma_1$ and $\sigma_2$ are the analog of the action $D$ in the Prisoner's Dilemma,
while the actions $\tau_1$ and $\tau_2$ are the analog of the action $C$.

Denote by $\sigma^*_i$ the course of action of player $i$ in which he always plays $\sigma_i$.
Denote by $GT^*_i$ the course of action of player $i$ in which he plays $\tau_i$ in the first stage,
and in every subsequent stage he plays $\tau_i$ if player $j$ played $\tau_j$ in all previous stages,
and plays $\sigma_i$ otherwise.

The definition of a conditional grim trigger strategy is analogous to Definition \ref{grimdef},
with the course of action $\sigma_i^*$ replacing the ``always defect'' course of action $D_i^*$:

\begin{definition}
Let $K_i \subseteq \Omega$ be a $\Sigma_i$-measurable event.
\emph{A conditional grim trigger strategy} for player $i$ with cooperation region $K_i$,
is the strategy $\eta_i^*(K_i)$ defined as follows:
\begin{equation}
\eta^*(K_i \mid \omega) = \left\{
\begin{array}{lll}
GT^*_i(\omega) & & \omega \in K_i,\\
\sigma^*_i(\omega) & \ \ \ \ \ & \omega\not\in K_i.
\end{array}
\right.
\end{equation}
\end{definition}

Cooperation events are defined similarly to Definition \ref{defin cooperation}.
In the complete information case,
there are thresholds
$\lambda_1^0,\lambda_2^0$ such that the players can agree to cooperate (and play according to $\tau$)
if and only if $\lambda_i(\omega)\ge \lambda_i^0$ for every $\omega \in \Omega$, where
$$\lambda_i^0:=\min\left\{\lambda_i\mid
\frac{u_i(\tau)}{1-\lambda_i}-\left(u_i(\sigma_i',\tau_j)+u_i(\sigma)\frac{\lambda_i}{1-\lambda_i}\right)\ge0
\;\;\forall \sigma_i'\ne\tau_i\right\}.$$ Denote
$\Lambda_i=\{\omega\in\Omega\mid\lambda_i(\omega)\ge\lambda_i^0\}$.
Then in the complete information case, Theorem \ref{theorem complete} holds:
$(K_1,K_2)$ are cooperation events if and only if $K_1=K_2\subseteq \Lambda$,
where $\Lambda:=\Lambda_1\cap\Lambda_2$.

The analog of Theorem \ref{conditionprop} in the general setup is the following.
\begin{theorem}
\label{generalconditionprop} Let $K_i\subseteq\Omega$ be a $\Sigma_i$-measurable event for each $i=1,2$.
The pair $(K_1,K_2)$ is a pair of cooperation events if and only if, for each $i=1,2$,
\begin{enumerate}
\item[(a)] $K_i\subseteq\Lambda_i$,
\item[(b)] $P_i(K_j\mid\omega)\ge f_i(\omega)$ for every $\omega\in K_i$, and
\item[(c)] $P_i(K_j\mid\omega)\le g_i(\omega)$ for every $\omega\notin K_i$,
\end{enumerate}
for the $\Sigma_i$-measurable functions%
\footnote{By convention, the infimum over an empty set is 1 and the supremum over an empty set is 0.}
\begin{equation}
\label{equ fi}
\begin{array}{ll}
      f_i(\omega):=
      \max_{\sigma_i'\in F_i}
      \frac{u_i(\sigma_i',\sigma_j)-u_i(\tau_i,\sigma_j)}{\left(\frac{u_i(\tau)}{1-\lambda_i(\omega)}-(u_i(\sigma_i',\tau_j)
      +u_i(\sigma)\frac{\lambda_i(\omega)}{1-\lambda_i(\omega)})\right)+(u_i(\sigma_i',\sigma_j)-u_i(\tau_i,\sigma_j))},
\end{array}
\end{equation}
and
$$
g_i(\omega):=\min\{g_i^1,g_i^2(\omega),g_i^3(\omega)\},
$$
where $\sigma_i'$ is an action of player $i$, and
$F_i=\{\sigma_i' \in A_i\mid u_i(\tau_i,\sigma_j)<u_i(\sigma_i',\sigma_j)\}$.
The functions $g_i^1$, $g_i^2(\omega)$ and $g_i^3(\omega)$ reflect several types of deviations, and are defined by
$$g_i^1:=\min_{\sigma_i'\in H_i^1}
        \frac{u_i(\sigma)-u_i(\sigma_i',\sigma_j)}
        {(u_i(\sigma)-u_i(\sigma_i',\sigma_j))+(u_i(\sigma_i',\tau_j)-u_i(\sigma_i,\tau_j))}$$
where $H_i^1:=\left\{\sigma_i'\in A_i \setminus\{\sigma_i,\tau_i\}\mid u_i(\sigma_i,\tau_j)<u_i(\sigma_i',\tau_j)\right\}$,
$$
        g_i^2(\omega):=
        \frac{u_i(\sigma)-u_i(\tau_i,\sigma_j)}{(u_i(\sigma)-u_i(\tau_i,\sigma_j))+
        \left(\frac{u_i(\tau)}{1-\lambda_i(\omega)}-(u_i(\sigma_i,\tau_j)
        +u_i(\sigma)\frac{\lambda_i(\omega)}{1-\lambda_i(\omega)})\right)}
$$
whenever $u_i(\sigma_i,\tau_j)+u_i(\sigma)\left(\frac{\lambda_i(\omega)}{1-\lambda_i(\omega)}\right)
<\left(\frac{u_i(\tau)}{1-\lambda_i(\omega)}\right)$, and $g_i^2(\omega):=1$ otherwise, and
$$
\begin{array}{l}
        g_i^3(\omega):=
        \min_{\sigma_i'\in H_i^3(\omega)}
        \frac{u_i(\sigma)-u_i(\tau_i,\sigma_j)}
        {(u_i(\sigma)-u_i(\tau_i,\sigma_j))+(u_i(\tau)-u_i(\sigma_i,\tau_j)+
        (u_i(\sigma_i',\tau_j)-u_i(\sigma))(\lambda_i(\omega))},
\end{array}
        $$
where
$$
        H_i^3(\omega):=\left\{\sigma_i'\in A_i \setminus \{\tau_i\}\mid
        u_i(\tau)-u_i(\sigma_i,\tau_j)+(u_i(\sigma_i',\tau_j)-u_i(\sigma))\lambda_i(\omega)<0\right\}.
$$
\end{theorem}
It is worth noting that the only difference between Theorem \ref{conditionprop} and Theorem \ref{generalconditionprop}
is in condition (c): in the general case, player $i$ $(1-g_i)$-believes that the event $K_j$ does not hold
at every $\omega \not\in K_i$,
and not $(1-f_i)$-believes that this event does not hold.
The reason for this difference is that in general games players have more options to deviate than in the Prisoner's Dilemma,
and therefore additional constraints affect the equilibrium conditions.

As in Section \ref{beliefsec}, the
conditions in Theorem \ref{generalconditionprop} are equivalent to
the following condition, which links the theorem to the concept of
$f$-common-belief: each player either $f$-believes that $K_1\cap
K_2$ is a common-$f$-belief or $(1-g)$-believes that the event
``$K_1\cap K_2$ is not a common-$f$-belief'' is a
common-$(1-g)$-belief (see Maor (2010) for details).

The proof of Theorem \ref{generalconditionprop} follows the lines of the proof of Theorem \ref{conditionprop};
the interested reader is referred to Maor (2010) for the complete proof.

If $\Gamma$ is not a
$2\times2$ game, there may be $\omega\in\Omega$ such
that $f_i(\omega)>g_i(\omega)$.
In this case
$(B_1^{f_1}(D^f(\Lambda)),B_2^{f_2}(D^f(\Lambda)))$ may not be a pair of cooperation events.
This point is illustrated in Examples \ref{example5} and \ref{example6} below.
In these two examples we consider the following $3 \times 3$ repeated game:

\begin{picture}(265,90)(15,-20)
\put( 15,8){$C$}
\put( 15,28){$D$}
\put( 80,50){$D$}
\put(160,50){$C$}
\put(15,-12){$N$}
\put(240,50){$N$}
\put(40,-20){\numbercellong{$0,\cdot$}{}}
\put( 40, 0){\numbercellong{$0,4$}{}}
\put(40,20){\numbercellong{$1,1$}{}}
\put(120,-20){\numbercellong{$a,\cdot$}{}}
\put(120,0){\numbercellong{$3,3$}{}}
\put(120,20){\numbercellong{$4,0$}{}}
\put(200,-20){\numbercellong{$\cdot,\cdot$}{}}
\put(200,0){\numbercellong{$\cdot,a$}{}}
\put(200,20){\numbercellong{$\cdot,0$}{}}
\end{picture}
\newline
where $a>4$; payoffs that are not indicated in the matrix can be arbitrary. Here
$\sigma=(D,D)$, $\tau=(C,C)$, and
$\lambda_1^0=\lambda_2^0=\frac{a-3}{a-1}$. In both examples,
$f_i(\omega)=g_i^2(\omega)=\frac{1-\lambda_i(\omega)}{2\lambda_i(\omega)}$,
$g_i^1=\frac{1}{a-3}$, and
$g_i^3(\omega)=\frac{1}{(a-1)\lambda_i(\omega)}$.

In the following example $B_i^{f_i}(D^f(\Lambda))\ne\emptyset$ for
$i=1,2$, but these sets do not satisfy the conditions in Theorem
\ref{generalconditionprop}. Moreover, we prove that in this
example there are no non-trivial pairs of cooperation events.
For simplicity the example uses an infinite state space. A similar construction is possible with a finite state space.

\begin{example}
\label{example5}

Let $a=6$.
The set of states of the world is $\Omega = (0,1)^2$,
and the beliefs of the players are derived from a common prior, the uniform distribution over $\Omega$.
One can verify that
$\Lambda_1=[3/5,1)\times(0,1)$ and $\Lambda_2=(0,1)\times[3/5,1)$.
Because $f_i(\omega)\le
f_i(3/5)={1\over3}<1-3/5=P_i(\Lambda_j\mid\omega)$ for every
$\omega\in\Lambda_i$,
it follows that $D_i^1(\Lambda_1,\Lambda_2)=\Lambda_i$, and
therefore $B_i^{f_i}(D^f(\Lambda_1,\Lambda_2))=D_i^f(\Lambda_1,\Lambda_2)=\Lambda_i$. But, for
$\omega\notin\Lambda_i$,
$P_i(\lambda_j\ge3/5 \mid\omega)=2/5>{1\over3}=g_i^1\ge g_i(\omega)$, and
therefore, in contrast to previous examples, the second condition in Theorem
\ref{generalconditionprop} does not hold, and
$(B_1^{f_1}(D^f(\Lambda)),B_2^{f_2}(D^f(\Lambda))))$ is not a pair of cooperation events.

We now argue that the only pair of cooperation events is $(\emptyset,\emptyset)$.
Suppose to the contrary that there is
a pair of non-empty cooperation events $(K_1,K_2)$. For $i=1,2$ denote
$\lambda_i^*:=\inf\{\lambda_i(\omega)\mid\omega\in K_i\}$. Since
$K_i\subseteq\Lambda_i$, $\lambda_i^*\ge3/5$. Note that
$P_i(K_j\mid\omega)$ is independent of the state of the world
$\omega$; denote this quantity by $P_i(K_j)$. From the first inequality
of Theorem \ref{generalconditionprop}, we have $P_i(K_j)\ge
f_i(\omega)$ for every $\omega\in K_i$, and since $f_i$ is
continuous, $P_i(K_j)\ge f_i(\lambda_i^*)$. We now argue that if
$\lambda_i(\omega)>\lambda_i^*$, then $\omega\in K_i$. Otherwise
we deduce from the second inequality of Theorem
\ref{generalconditionprop}, that $P_i(K_j)\le g_i(\omega)\le
g_i^2(\omega)=f_i(\omega)$.
Since $\lambda_i(\omega)>\lambda_i^*$ it follows that $f_i(\omega)<f_i(\lambda_i^*)\le
P_i(K_j)$, a contradiction. Therefore, we have that $P_i(K_j)=1-\lambda_j^*$.

Next, we argue that $P_i(K_j)=f_i(\lambda_i^*)$. Otherwise
$P_i(K_j)>f_i(\lambda_i^*)$, and there is a state of the world
$\omega\in\Omega$ such that
$P_i(K_j)>f_i(\omega)>f_i(\lambda_i^*)$. From the definition of
$\lambda_i^*$ we have $\omega\notin K_i$, but then we should have
$P_i(K_j)\le g_i^2(\omega)=f_i(\omega)$, a contradiction. We
conclude that $1-\lambda_j^*=f_i(\lambda_i^*)$ for $i=1,2$.
Because $f_i(\lambda_i)=\frac{1-\lambda_i}{2\lambda_i}$, we have
that
$\frac{1-\lambda_2^*}{2\lambda_2^*}=1-\lambda_1^*=2\lambda_1^*(1-\lambda_2^*)$,
or equivalently, $\lambda_1^*\lambda_2^*={1\over4}$.
But $\lambda_i^*\ge3/5$ for $i=1,2$, a contradiction.
\end{example}

In the following example $B_i^{f_i}(D^f(\Lambda))\ne\emptyset$ and
does not satisfy the conditions in Theorem
\ref{generalconditionprop}, but for a smaller non-empty set, the
conditions hold, and therefore define a Bayesian equilibrium.

\begin{example}
\label{example6}

Let $a=5$, and $\Omega=\{{1\over4},{1\over2},{3\over4}\}^2$.
The beliefs of the players are derived from a common prior, which is the uniform distribution over $\Omega$.
Here $\lambda_i^0={1\over2}$, and therefore
$\Lambda_i=\{\lambda_i(\cdot)\ge{1\over2}\}$. As in the Example
\ref{example5},
$B_i^{f_i}(D^f(\Lambda))=D_i^f(\Lambda_i,\Lambda_j)=\Lambda_i$, but the
second condition of Proposition \ref{generalconditionprop} does
not hold because of $g_i^1$:
$P_1(\Lambda_2\mid({1\over4},\lambda_2))=2/3>{1\over2}=g_i^1$.

For $C_i=\{\lambda_i={3\over4}\}$ we get $D_i^f(C_i,C_j)=C_i$, and  one can verify that the
second condition of Proposition \ref{generalconditionprop} does
hold. Therefore,
$$
\eta_i^*(\omega)=
\left\{\begin{array}{ll}
\tau_i^*& \lambda_i={3\over4},\\
\sigma_i& \lambda_i={1\over4},{1\over2},
\end{array}
\right.
$$
defines a Bayesian equilibrium.
\end{example}

\subsection{``Almost'' Complete Information}
\label{section almost}

Monderer and Samet (1989) use the notion of common $p$-belief to measure approximate common knowledge.
A natural question is then, whether when there is common $(1-\ep)$-belief
regarding the discount factor for sufficiently small $\ep > 0$,
there is a pair of cooperation events that is close to $(\Lambda,\Lambda)$,
the largest pair of cooperation events in the complete information case.

In this section we provide two natural definitions for the notion ``almost complete information'';
in one the answer to the question we posed is negative, in the other it is positive.
We present the results without proofs;
the interested reader is referred to
Maor (2010)
for additional examples and for the complete proofs,
which are basically applications of Theorems \ref{conditionprop} and \ref{generalconditionprop} and results from Monderer and Samet (1989).
The discussion in this section is related to Monderer and Samet (1996),
who show that the concept of common repeated $p$-belief is related to
the continuity of the Nash equilibrium correspondence
relative to the information structure (see also Einy et al. (2008)).

Example \ref{prisonerex4} shows that when each player approximately knows the
other player's discount factor,
the cooperation events may be
significantly different from $\Lambda$;
this is true even if we consider only
$\delta$-equilibrium for some $\delta>0$.
Other concepts of perturbation of the complete information case are also possible, and they may yield different results.
For example, the structure theorem in Weinstein and Yildiz (2007)
shows that there are arbitrary small perturbations of the product topology of beliefs that make no conditional grim trigger strategy rationalizable.

Following Monderer and Samet (1989), we can suggest the following
definition of almost complete information.
\begin{definition}
\label{defin almost comp MS} Assume a common prior $P$ over the
set of states of the world. Let $\epsilon,\delta>0$. We say that the
discount factors are \emph{almost complete information} with
respect to $\epsilon$ and $\delta$, if the set of states of the
world in which the true discount factors are
common-$(1-\epsilon)$-belief has probability at least $1-\delta$.
\end{definition}
In other words, this definition means that in most states of the world, the true state of nature in a common $p$-belief for a high $p$,
that is, the true state of nature is known with a high probability, that fact is known with high probability, etc.

From Theorem B in Monderer and Samet (1989) we deduce that if the
number of states of nature is finite, there
are a strategy profile $\eta$ and an event $\Omega'$ with
probability at least $(1-2\epsilon)(1-\delta)$, such that (a)
$\eta(\omega)=\eta^*(\Lambda,\Lambda)(\omega)$ for every
$\omega\in\Omega'$, and (b) $\eta$ is an $\epsilon'$-equilibrium
for $\epsilon'>4M\epsilon$,
where
$M$ is
the maximal absolute value of the payoff that a player can obtain in any state of the world:
\[ M := \sup_{\omega \in \Omega} \max_{a \in A_1 \times A_2} \max_{i=1,2}\frac{u_i(a)}{1-\lambda_i(\omega)} < \infty. \]
In other words, there is an
$\epsilon'$-equilibrium that coincides over a large set
with the conditional grim trigger equilibrium of maximum
cooperation in the complete information case.

In the Prisoner's Dilemma and other $2\times2$ games, this
profile $\eta$ can indeed be a conditional grim trigger profile,
as the following proposition states.
In the proposition's statement,
$f^\ep_i$ is the function defined in (\ref{equ fi}) after subtracting $\ep$ from the numerator.
\begin{proposition}
\label{cor 2 action almost} Suppose each player has two actions.
Let $\epsilon>0$ and $\delta>0$ be given, and denote
$M_0:=2\max_{i=1,2}(u_i(\sigma)-u_i(\tau_i,\sigma_j))$. Assume that
the discount factors are almost complete information with respect
to $\epsilon$ and $\delta$. Then, for every $\epsilon'\ge
M_0\epsilon$, the strategy profile
$\eta^*(B_1^{f^{\epsilon'}_1}(D^{f^{\epsilon'}}(\Lambda)),B_2^{f^{\epsilon'}_2}(D^{f^{\epsilon'}}(\Lambda)))$
is an $\epsilon'$-equilibrium, and $P(\Lambda\setminus
D^{f^{\epsilon'}}(\Lambda))<\delta$.
\end{proposition}

For games where a player has more than two actions, the profile
$\eta$ may not be a conditional grim trigger strategy profile.
Moreover, there may be no conditional grim trigger
$\epsilon'$-equilibria whatsoever, even in the case of a finite state space (see Example 8.3 in Maor
(2010)). This happens because this definition of almost
complete information allows the existence of ``problematic" states
that occur with small probability, in which the beliefs of the players can
be far from complete information, while conditions (b) and
(c) in Theorem \ref{generalconditionprop} pose demands on all
states of the world.

Therefore, with conditional grim trigger strategies in games that are larger than
$2\times2$, to obtain an $\epsilon'$-equilibrium
that coincides over a large set with the
conditional grim trigger equilibrium of maximum cooperation in the
complete information case,  we need a stronger concept of almost complete
information, which is ``almost complete'' in \emph{all} states of the world, and not merely over a large set of states of the world.

\begin{definition}
\label{defin almost comp strong} Let $\epsilon>0$. We say that the
discount factors are \emph{almost complete information} with
respect to $\epsilon$, if at every state of the world $\omega$,
each player $(1-\epsilon)$-believes that some state of
nature is common-$(1-\epsilon)$-belief in $\omega$.
\end{definition}

When the information is almost complete according to this
definition, we can point at an $\epsilon'$-equilibrium in conditional grim trigger strategies:
\begin{proposition}
\label{prop almost complete} Let $\epsilon>0$. Assume that the
discount factors are almost complete information with respect to
$\epsilon$ (according to Definition \ref{defin almost comp
strong}). Then the strategy profile
$\eta^*(B_1^{1-\epsilon}(D^{1-\epsilon}(\Lambda)),B_2^{1-\epsilon}(D^{1-\epsilon}(\Lambda)))$
is an $\epsilon'$-equilibrium, for every $\epsilon'>M\epsilon$,
where $M>0$ is a constant, independent of the information
structure and of $\epsilon$.
\end{proposition}

While there can be almost complete information on the discount factors
(according to Definition \ref{defin almost comp strong}) even if there is no common prior,
if there is a common prior,
then this concept is stronger than the
concept in Definition \ref{defin almost comp MS}. In this case,
under a common prior, the $\epsilon'$-equilibrium in conditional grim trigger strategies
described above indeed coincides over a
large set with the conditional grim trigger equilibrium of maximum
cooperation in the complete information case, as stated by the following
proposition. In other words, under a stronger concept of
almost complete information, we have a result that is quite
similar to the one in Theorem B of Monderer and Samet (1989), but
with a conditional grim trigger profile that is defined
explicitly for every state of the world.

\begin{proposition}
\label{cor almost comp} Let $\epsilon>0$. Assume that (a) the
beliefs of the players are derived from a common prior $P$; (b) the
discount factors are almost complete information with respect to
$\epsilon$ (according to Definition \ref{defin almost comp
strong}); and (c) the number of states of nature is finite or
countable. Then (i) the strategy profile
$\eta^*(B_1^{1-\epsilon}(D^{1-\epsilon}(\Lambda)),B_2^{1-\epsilon}(D^{1-\epsilon}(\Lambda)))$
is an $\epsilon'$-equilibrium, for every $\epsilon'>M_1\epsilon$,
where $M_1>0$ is a constant, independent of the players' beliefs and of $\epsilon$; (ii) under this strategy profile, both
players will cooperate if and only if $\omega\in
D^{1-\epsilon}(\Lambda)$; and (iii) $P(\Lambda\setminus
D^{1-\epsilon}(\Lambda))<3\epsilon$.
\end{proposition}

\subsection{Further Generalizations}
Theorem \ref{conditionprop} can also be generalized for other
cases, with the necessary adjustments to the functions $f_i$ and
$g_i$ and the sets $\Lambda_i$.
One generalization is for the case of incomplete information regarding one's own discount
factor.
A second generalization concerns games with
more than two players.
In these cases, using similar ideas one can obtain analog results regarding sufficient conditions for equilibria in
conditional grim trigger strategies.
Because the information structure and the game structure are more
complicated in multiplayer games, further assumptions should be
made. E.g., in the case of more than two players, to treat all the other
players within the framework of one function $f$ one needs to take
some ``worst case scenario".  This makes the
conditions only sufficient and not necessary. See Maor (2010) for
details.

Lastly, the theorem can be extended to general games with incomplete
information.
Let $G=\left(N,(S,\mathcal{S}),\Pi,(A_i)_{i\in N},(u_i)_{i\in
N}\right)$ be a general two-player Bayesian game. To avoid
measurability problems, assume that the set of states of the world is
finite or countable. Assume that there is a course of action
profile $\sigma^*=(\sigma_1^*,\sigma_2^*)$ that is an equilibrium for all states of
nature when information is complete,
and that there is another course of action profile,
$\tau^*=(\tau_1^*,\tau_2^*)$ that is an equilibrium only in some
states of nature when the information is complete.
Suppose that the supports of $\tau_i^*$ and
$\sigma_i^*$ are disjoint in all states of the world; that is, it
is discernable whether player $i$ plays $\sigma_i^*$ or
$\tau_i^*$. A strategy of player $i$ is an $i$-measurable function
that assigns a course of action of player
$i$ to each state of the world.

Let $\Lambda_i\subseteq\Omega$ be the event ``player $i$ believes
that he cannot benefit by deviating from the profile $\tau^*$''.
That is, for every $\omega\in\Lambda_i$ and every course of action
$\sigma_i'$ of player $i$, $E_i(u_i(\tau^*)\mid\omega)\ge
E_i(u_i(\sigma_i',\tau_j^*)\mid\omega)$.

\begin{theorem}
\label{generalgamesprop} In the game $G$ there exist
$i$-measurable functions $0\le f_i,g_i\le1$, $i=1,2$, such that if
$K_1\subseteq\Lambda_1$ and $K_2\subseteq\Lambda_2$, the strategy
profile $\eta^*(K_1,K_2)=(\eta_1^*(K_1),\eta_2^*(K_2))$ is a
Bayesian equilibrium if and only if for $i=1,2$,
\begin{enumerate}
  \item $P_i(K_j\mid\omega)\ge f_i(\omega)$ for every $\omega\in K_i$, and
  \item $P_i(K_j\mid\omega)\le g_i(\omega)$ for every $\omega\notin K_i$.
\end{enumerate}
or, in other words, if and only if
\begin{enumerate}
  \item $K_i \subset B_i^{f_i}(K_j)$, and
  \item $K_i^c \subset B_i^{1-g_i}(K_j^c)$.
\end{enumerate}
\end{theorem}

There are several differences between Theorem \ref{generalgamesprop} and Theorem \ref{generalconditionprop}.
First, here we assume that $K_1\subseteq\Lambda_1$ and $K_2\subseteq\Lambda_2$ whereas in Theorem \ref{generalconditionprop} we showed that
   it was necessary for $\eta^*(K_1,K_2)$ to be a Bayesian equilibrium.
   This weakened result allows us to drop the assumption that $\tau_i^*$ is not a best response to $\sigma_j^*$,
   and only requires that it is discernable whether player $i$ plays $\sigma_i^*$ or $\tau_i^*$ (disjoint supports).%
   \footnote{Still, as before, if we assume that $\tau_i^*$ is not a best response to $\sigma_j^*$ and if the payoffs in $G$ are bounded, then $f_i>0$.}
Second, we do not assume that the realized payoffs are observed,
  so that all a player knows is his expected payoff based on his information on the states of nature.
Third, we do not assume that the payoffs when $\tau^*$ is played are higher than when $\sigma^*$ is played.

For further generalizations and complete the reader is referred to Maor
(2010).

%

\appendix

\section{Appendix: Proof of Theorem \ref{conditionprop}}
\label{appendix proof}

Because a strategy of player $i$ is a $\Sigma_i$-measurable function,
a necessary condition for $\eta^*_i(K_i)$ to be a strategy is that $K_i$ is $\Sigma_i$-measurable.

For each $i \in \{1,2\}$, let $K_i$ be a $\Sigma_i$-measurable event.
We will check when a deviation from $\eta^*(K_1,K_2)$ can be profitable.

Assume first that $\omega\in K_i$. In this case player $i$ is
supposed to follow the grim trigger course of action $GT^*$.

If $\omega \not\in K_j$, player $j$ will play $D$ all along the game, regardless of the play of player $i$.
If $\omega \in K_j$, player $j$ will follow the grim trigger course of action $GT^*$.
It then follows that if player $i$ deviates and plays $D$ at stage $k$,
then he receives the highest payoff if he continues to play $D$ after stage $k$.
Because payoffs are discounted, if a deviation to $D$ at stage $k$ is profitable to player $i$,
then a deviation at stage 1 provides a higher profit.
Therefore, the best possible deviation of player $i$ is to the course of action $D^*$.
This deviation is not profitable if and only if
\begin{equation}
\label{equ89}
\gamma_i(\eta^*(K_i,K_j)\mid\omega) \geq \gamma_i(D^*_i,\eta^*_j(K_i,K_j)\mid\omega),
\end{equation}
where $D^*_i$ is player $i$'s strategy in which he always plays $D$.
One can verify that
\begin{equation}
\gamma_i(\eta^*(K_i,K_j)\mid\omega)=
P_i(K_j\mid\omega)\left(\frac{3}{1-\lambda_i(\omega)}\right)
+(1-P_i(K_j\mid\omega))\frac{\lambda_i(\omega)}{1-\lambda_i(\omega)}.
\label{equ90}
\end{equation}
Indeed, according to player $i$'s belief, with probability
$P_i(K_j\mid\omega)$ player $j$ plays $GT^*$, so they will play
$(C,C)$ at every stage and his payoff will be
$\frac{3}{1-\lambda_i(\omega)}$, and with probability
$(1-P_i(K_j\mid\omega))$ player $j$ plays $D^*$, so in the first
stage he will get 0, and afterwards the
players will play $(D,D)$; therefore in this case player $i$'s payoff is
$\frac{\lambda_i(\omega)}{1-\lambda_i(\omega)}$.
Similarly,
\begin{equation}
\label{equ91}
\begin{array}{ll}
      \gamma_i(D_i^*,\eta_j^*(K_j)\mid\omega)=
      P_i(K_j\mid\omega)\left(4+\frac{\lambda_i(\omega)}{1-\lambda_i(\omega)}\right)
      +(1-P_i(K_j\mid\omega))\left(\frac{1}{1-\lambda_i(\omega)}\right).
\end{array}
\end{equation}
Plugging in (\ref{equ90}) and (\ref{equ91}) in (\ref{equ89}) yields that $D^*_i$ is not a profitable deviation for player $i$
if and only if $P_i(K_j\mid\omega)\ge f_i(\omega)$, which is condition (b). Since $f_i(\omega)>1$ for $\omega \notin \Lambda_i$,
this also implies that $K_i \subseteq \Lambda_i$, which is condition (a).

Now assume that $\omega\notin K_i$.
In this case under $\eta^*_i(K_1,K_2)$ player $i$ follows the course of action $D^*$.
If player $i$ follows $\eta^*_i(K_1,K_2)$ at the first stage, then from the second stage on player $j$ will play $D$,
and then the best reply of player $i$ is to play $D$.

If $\omega\not\in K_j$, player $j$ follows $D^*$,
and $\eta^*_i(K_1,K_2)$ is player $i$'s best response.
If $\omega \in K_j$, player $j$ follows $GT^*$,
and therefore if there is a profitable deviation to player $i$,
he must deviate at the first stage.
Thus, to check whether $\eta^*(K_1,K_2)$ is a Bayesian equilibrium we need to check whether
$\gamma_i(\eta^*(K_1,K_2)) \ge \gamma_i(GT^*_i,\eta^*_j(K_1,K_2))$,
where $GT^*_i$ is the grim trigger course of action of player $i$.

One can verify that
$$
\gamma_i(\eta^*(K_1,K_2)\mid\omega)=(1-P_i(K_j\mid\omega))\frac{1}{1-\lambda_i(\omega)}+
P_i(K_j\mid\omega)\left(4+\frac{\lambda_i}{1-\lambda_i}\mid\omega\right)
$$
and
$$
        \gamma_i(C_i^*,\eta_j^*(K_j)\mid\omega)=P_i(K_j\mid\omega)\frac{3}{1-\lambda_i(\omega)}+
        (1-P_i(K_j\mid\omega))\left(\frac{\lambda_i(\omega)}{1-\lambda_i(\omega)}\right).
$$
The inequality $\gamma_i(\eta^*(K_1,K_2)) \geq \gamma_i(GT^*_i,\eta^*_j(K_1,K_2))$ then reduces to
$P_i(K_j\mid\omega)\le f_i(\omega)$, and condition (c) is obtained.

Finally we need to verify
that $f_i$ is $\Sigma_i$-measurable. This function is a rational function of
$\lambda_i$, and therefore it is a Borel function of $\lambda_i$.
Since $\Sigma_i$ contains the sets $\{\omega\mid\lambda_i(\omega)\in B\}$,
for every open set $B\subseteq[0,1)$, $f_i$, as a function of
$\omega$, is indeed $\Sigma_i$-measurable.

\end{document}